\documentclass[15 pt]{amsart}
\usepackage{amsfonts, amsmath, amssymb, amscd, amsbsy, amscd, amsgen, amsopn, amstext, amsxtra,cite}
\usepackage[latin5]{inputenc}
\usepackage{xcolor}
\usepackage{tikz}

\newtheorem{thm}{Theorem}[section]

\newtheorem{cor}[thm]{Corollary}

\newtheorem{lem}[thm]{Lemma}

\newtheorem{pro}[thm]{Proposition}
\theoremstyle{definition}
\newtheorem{defn}[thm]{Definition}
\theoremstyle{question}
\newtheorem{que}[thm]{Question}
\theoremstyle{remark}
\newtheorem{rem}[thm]{Remark}
\theoremstyle{example}

\theoremstyle{note}

\newcommand{\J}{\mathcal{J}}

\def\PAut{\mathrm{PAut}}

\begin{document}

\title{On the  binary linear constant weight codes and their autormorphism groups }

\author{ Murat Altunbulak, Fatma Altunbulak  Aksu  }

\address{Department of Mathematics, Mimar Sinan Fine Arts University, \c{S}\.I\c{s}l\.I, \.Istanbul, Turkey}

\email{fatma.altunbulak@msgsu.edu.tr, }

\address{Department of Mathematics, Dokuz Eyl\"ul University, Buca, \.Izm\.Ir, Turkey}

\email{murat.altunbulak@deu.edu.tr,}


\thanks{}


\date{\today}

\dedicatory{}

\commby{}

\begin{abstract} We give a characterization for the binary linear constant weight codes by using the symmetric difference of the supports of the codewords. This characterization gives a correspondence between the set of binary linear constant weight codes and the set of partitions  for the union  of supports of the codewords. By using this correspondence, we present a formula for the order of the automorphism group  of a binary linear constant weight code in terms of its parameters. This formula is a key step to determine more algebraic structures on constant weight codes with given parameters.  Bonisoli [Bonisoli, A.: 
Every equidistant linear code is a sequence of dual Hamming codes. Ars Combinatoria 18, 181--186 (1984)] proves  that the $q$-ary linear constant weight codes with the same parameters are  equivalent (for the binary case permutation equivalent). We also give an alternative proof for Bonisoli's theorem by presenting  an explicit permutation on symmetric difference of the supports of the codewords which gives the permutation equivalence between the binary linear constant weight codes.
\keywords{constant weight code, permutation automorphism, symmetric difference, support of a codeword
}
\end{abstract}
\maketitle
\section{Introduction} \label{Sec:1}
A binary linear code $C$ of length $n$ is  a subspace of the vector space $\mathbb{F}_2^n$.  The support of a non-zero codeword $c=(c_1,c_2,\ldots,c_n)$ in $C$ is  the set of its non-zero coordinate positions, i.e. ${\rm Supp}(c)=\{i\in\{1,2,\ldots,n\}\ : \ c_i=1\}$. The Hamming weight  of $c$  is the number of elements of the support of $c$. A binary linear code is called constant weight code if every non-zero codeword has the same Hamming weight.  There is an increasing interest on binary constant weight codes (see \cite{AgrellVardyZeger,BrouwerShearerSloaneSmith,CheeLing,EtzionSchwarz,EtzionVardy,Gashkov,Hou,MacWilliamsSloane,NguyenGyorfiMassey,Ostergard,SmithHughesPerkins}).   The constant weight codes are building blocks of the Hamming scheme \cite{EtzionSchwarz}. This specific family of codes have many applications \cite{NguyenGyorfiMassey,BrouwerShearerSloaneSmith}. A relation between the design theory and the constant weight codes are investigated in \cite{Ostergard}. There in,   the supports of the codewords in a constant weight code form a special packing in the design theory. Inspired by these relations, we  consider on the supports of the codewords in a constant weight code and then by using the repeated symmetric difference of the supports of the codewords, we  give a characterization of the binary linear constant weight codes (see Proposition \ref{intersection of Supports} and Theorem \ref{converse}). This characterization  mainly depends on partitions for the union of supports of basis  elements of the given code. By using such partitions for the union of supports, we give an explicit  construction for the binary linear constant weight codes (see Theorem \ref{construction}).

One of the important problems in coding theory is the code equivalence problem.  The symmetric group $S_n$   acts on $\mathbb{F}_2^n$ by permuting the coordinates of elements.  Two linear codes $C,D$ of length $n$ are said to be permutation equivalent if there exists a permutation $\sigma\in S_n$ such that $\sigma(C)=D$. In most cases, it is quite difficult to decide that two given linear codes are permutation equivalent. It is also very important to find the explicit permutation which gives the permutation equivalence between these linear codes. In this paper,  we   give a simple proof for the permutation equivalence of the binary linear constant weight codes ( see Theorem \ref{Perm_Equiv}) by writing an explicit permutation.  We should note that in \cite{Bonisoli},  Bonisoli proves the equivalence of the $q$-ary constant weight codes of the given parameters by an induction argument (see the details in Remark \ref{explanation}).

 For a linear code  $C$, the subgroup $\{\sigma\in S_n| \sigma(C)=C\}$ of $S_n$ is called permutation automorphism group of $C$ and denoted by $\PAut(C)$.  For a given linear code $C$, whenever $\PAut(C)$ is non-trivial, we get many algebraic properties such as being
$\mathbb{F}_2 G$-module  where $G$ is a non-trivial subgroup of $\PAut(C)$ or being a group code or being a quasi group code (see \cite{BernalRioSimon,Borello,BorelloWillems}). The characterization of binary linear constant weight codes in terms of partitions for the union of supports of codewords is a key step to calculate the order of the permutation automorphism group of the given code. By using this characterization, 
we give a formula for the order of the permutation automorphism group of a binary linear constant weight code. The formula is expressed in terms of the parameters of the code (see Theorem \ref{order_of_PAut 2}). For the order of $\PAut(C)$, such formulas in terms of parameters of the code is not common. 
As an application of this formula and the explicit construction in Theorem \ref{construction}, we calculate $\PAut(C)$ for some specific infinite families of constant weight codes. Moreover, we determine that for some specific parameters, there is no binary linear constant weight code which is also a group code.

The structure of the paper is as follows. In section 2, we give some basic properties of the repeated symmetric difference of some family of sets and then we relate these properties with supports of codewords in a given constant weight code. In section 3, we give a characterization for a given  constant weight code in terms of the symmetric difference of supports of codewords. As a corollary, we get lower bounds  for the weight and the length of a given   constant weight code in terms of its dimension. In section 4, we give a construction for constant weight codes of arbitrary dimension. In  section 5, we give a formula for the  order of permutation automorphism group of a given constant weight code in terms of its dimension, length and weight. We also present some applications of this formula. 

\section{Preliminaries} \label{Sec:2}
All codes in this paper are assumed to be binary linear codes. For a binary linear code $C$, the minimum weight $wt(C)$ of $C$ is defined as ${\rm min}\{wt(c)~| ~0\neq c\in C\}$.
A binary linear code $C$ is called constant weight code  if $wt(C)=wt(c)$ for any non-zero codeword $c\in C$.
By definition, all $1$-dimensional codes are constant weight code. So throughout the paper, when we write constant weight code, we mean a constant weight code of dimension at least $2$.
In most cases, whenever length is fixed,  it is not the case that we can find a constant weight code for any dimension. For example, there is no $3$-dimensional constant weight code of length $4$. There is also restriction about the weight of  a constant weight code of a given length. For example, if $C$ is a constant weight code of length $4$, then $wt(C)=2$  . So it is worthful to investigate the relations between the parameters (dimension, length, weight). In fact the relations on these parameters are studied in \cite{Bonisoli} by using Plotkin bound. Here we get the relations between these parameters as a corollary of a characterization of binary linear constant weight codes in terms of the symmetric difference of the supports of codewords.

Let $c_1,c_2 \in C$ be any two different non-zero codewords with supports $ A_1, A_2$ respectively. Suppose $| A_1|=w_1,| A_2|=w_2$ and $| A_1\cap  A_2|=m$.  We may write $ A_1 = \{i_1,i_2,\ldots,i_{w_1-m},j_1,\ldots,j_m\}$ and $ A_2 = \{l_1,l_2,\ldots,l_{w_2-m},j_1,\ldots,j_m\}$. As $C$ is a binary linear code, we have $c_1+c_2\neq 0$. Then $ A = {\rm Supp}(c_1+c_2) = \{i_1,i_2,\ldots,i_{w_1-m},l_1,l_2,\ldots,l_{w_2-m}\}$ which implies that $| A| = w_1+w_2-2m$. So for binary linear constant weight codes, we have the following simple observation on supports of the codewords.
\begin{lem}\label{intersectweight} Let $C$ be a constant weight code of weight $w$. If ${\rm dim}(C)\geq 2$, then for any different non-zero codewords $c_1, c_2$, we have $|{\rm Supp}(c_1)\cap {\rm Supp}(c_2)|=\frac{w}{2}$.
\end{lem}
\begin{proof} Let $c_1, c_2$ be two different  non-zero codewords in $C$. Then $wt(c_1)=wt(c_2)=w$. As $C$ is constant weight code, we have $wt(c_1+c_2)=w$. Let
$|{\rm Supp}(c_1)\cap {\rm Supp}(c_2)|=m $. Then $m$ is non-zero (otherwise we have $wt(c_1+c_2)=2w$) and we have   $wt(c_1+c_2)=2w-2m=w$. This implies that $m=\dfrac{w}{2}$.
\end{proof}

Hence we get a very well known fact on weights.
\begin{cor}\label{evenweight} If $C$ is a constant weight code of dimension $\geq 2$, then the weight of $C$ is even.
\end{cor}

Note that, we have ${\rm Supp}(c_1+c_2) =  A_1 \Delta  A_2$  whenever $c_1+c_2\neq 0$. Here $A_1 \Delta  A_2$ is  the symmetric difference of the supports of $c_1$ and $c_2$.   If we have $k$ linearly independent codewords $c_1,\ldots,c_k$, then the support of the sum $c_1+\dots +c_k$ is equal to the repeated symmetric difference of the supports of $c_1,\ldots,c_k$.
\subsection{Repeated Symmetric Differences}\label{Sec:3}
In this subsection, we investigate some properties of the repeated symmetric difference of  sets.
\begin{defn}\label{symdif} The repeated symmetric difference of  the sets $A_1,A_2,\ldots,A_k$ or the $k$-ary symmetric difference of $A_i$'s is the set $A_1\Delta A_2\Delta\cdots\Delta A_k$ of elements which are in an odd number of sets $A_i$'s, i.e.,
$$A_1\Delta A_2\Delta\cdots\Delta A_k = \{a\in \bigcup A_i \; : \; |\{i \, : \, a\in A_i\}| \mbox{ is odd}\}.$$
\end{defn}
\begin{defn}\label{AI} Let $A_1,A_2,\ldots,A_k$ be  non-empty sets. For a non-empty subset $I\subseteq [k]=\{1,2,\ldots,k\}$, we define the set
$$A^I = \left(\bigcap_{i\in I}A_i\right)\backslash\left(\bigcup_{j\notin I}A_j\right),$$
\end{defn}
Combining  Definitions \ref{symdif} and \ref{AI}, we get the following lemma.
\begin{lem}\label{oddindex} For non-empty sets $A_1,A_2,\ldots,A_k$ we have $$A_1\Delta A_2\Delta\cdots\Delta A_k = \bigcup_{\substack{I\subseteq[k],\\|I|~ is ~odd}} A^I.$$
\end{lem}
\begin{proof}
By definition, the set $A^I$ contains those elements in $\cup A_i $ which belong to only those $A_i$'s for $i$ in $I$. So if $|I|$ is odd, by definition of $k$-ary symmetric difference we have $A^I \subseteq A_1\Delta A_2\Delta\cdots\Delta A_k$, and hence we have
$$\bigcup_{\substack{I\subseteq[k],\\|I|~is ~odd}} A^I \subseteq A_1\Delta A_2\Delta\cdots\Delta A_k.$$
Conversely, if $x\in A_1\Delta A_2\Delta\cdots\Delta A_k$, then the number $|\{i \, : \, x\in A_i\}| $ is odd, by definition. So $x\in A^I$, where $I = \{i \, : \, x\in A_i\}$ and $|I|$ is odd. Therefore $x\in A^I\subseteq \bigcup_{\substack{I\subseteq[k],\\|I| ~is~ odd}} A^I $. Hence, we get that
$$A_1\Delta A_2\Delta\cdots\Delta A_k = \bigcup_{\substack{I\subseteq[k],\\|I|~is~odd}} A^I.$$
\end{proof}
If $I=\{i_1,i_2,\ldots,i_l\}$, then we   use the notation $A^{i_1i_2\ldots i_l}$ for $A^I$. The following lemma is an easy observation. For completeness, we give a simple proof for it.

\begin{lem}\label{omega} With the same notation in the definition, the following results hold

{\bf{(i)}} If $I\neq J$, then $ A^I\cap A^J=\emptyset$.

{\bf{(ii)}} For any $l\in [k]$, we have $ A_l = \bigcup_{l\in I, I \subseteq [k]} A^I$.

{\bf{(iii)}} $\{ A^I \ : \ A^I\neq\emptyset,\ I\subseteq [k], |I|\ge1\}$ form a partition for the set $\bigcup_{i=1}^k A_i$.
\end{lem}
\begin{proof}

{\bf{(i)}} Assume $I\neq J$.  Without loss of generality we can  assume that there exists $i_0 \in I$ such that $i_0\notin J$.  Then by definition, for any $x\in A^I$ we have $x\in A_i$ for all $i \in I$ and $x\notin A_j$ for all $j\notin I$. In particular, $x\in A_{i_0}$. Since $i_0\notin J$, by definition $x\notin A^J$. Hence, $A^I\cap A^J = \emptyset$.

{\bf{(ii)}} Since $A^I = \left(\bigcap_{i\in I}A_i\right)\backslash\left(\bigcup_{j\notin I}A_j\right)\subseteq \left(\bigcap_{i\in I}A_i\right)\subseteq A_i$ for all $i \in I$. If $l\in I$, then $ \bigcup_{\substack{l\in I, I\subseteq[k]}} A^I \subseteq A_l $. Conversely, if $x\in A_l$, then we have either $x\notin A_i$ for all $i\neq l$, or $x\in A_i$ for some $i \neq l$. As a result, we have either $x\in A^l$ or  $x\in A^{il}$. In any case we get, $x\in \bigcup_{\substack{l\in I, I\subseteq[k]\\}} A^I = A^l\sqcup A^{il}\sqcup\cdots$. Hence, we obtain $ A_l=\bigcup_{\substack{l\in I, I\subseteq[k]}} A^I$.

{\bf{(iii)}} Easily follows from $i)$ and $ii)$.
\end{proof}
\begin{pro} \label{intersections} For any $2\le l\le k$, we have
\begin{equation}\label{intersection in terms of AI}
A_{i_1}\cap A_{i_2}\cap\cdots\cap A_{i_l}= \bigsqcup_{\substack{I\subseteq[l],\\ \{i_1,i_2,\ldots,i_l\}\subseteq I}} A^I
\end{equation}
\end{pro}
\begin{proof}
If $\{i_1,i_2,\ldots,i_l\}\subseteq I$, then $A^I = \left(\bigcap_{i\in I}A_i\right)\backslash\left(\bigcup_{j\notin I}A_j\right) \subseteq \bigcap_{i\in I}A_i \subseteq A_{i_1}\cap A_{i_2}\cap\cdots\cap A_{i_l}.$ Therefore,
$$\bigsqcup_{\substack{I\subseteq[k],\\ \{i_1,i_2,\ldots,i_l\}\subseteq I}}A^I \subseteq A_{i_1}\cap A_{i_2}\cap\cdots\cap A_{i_l}.$$
Conversely, if $ x \in  A_{i_1}\cap A_{i_2}\cap\cdots\cap A_{i_l}$, then $x\in A_{i_j}$ for all $j=1,2,\ldots,l$. Now, we have two possible cases: either $x\notin A_r$ for any $r\notin \{i_1,i_2,\ldots,i_l\}$, or $x\in A_r$ for some $r\notin \{i_1,i_2,\ldots,i_l\}$. In the first case $x\in A^I$, for $I = \{i_1,i_2,\ldots,i_l\}$. For the latter case $x\in A^I$ for $I = \{i \ : \ x\in A_i\}$, which contains $\{i_1,i_2,\ldots,i_l\}$. So in any case, $x\in A^I$, where $\{i_1,i_2,\ldots,i_l\}\subseteq I$. Therefore,
$$A_{i_1}\cap A_{i_2}\cap\cdots\cap A_{i_l}\subseteq \bigsqcup_{\substack{I\subseteq[k],\\ \{i_1,i_2,\ldots,i_l\}\subseteq I}} A^I$$
Hence, we get the desired equality.
\end{proof}

Let $J = \{j_1,j_2,\ldots,j_l\}$ be a non-empty subset of $[k]$ and $\mathcal{A}=\{A_1,A_2,\ldots,A_k\}$. We denote the symmetric difference relative to $J$ as  $\Delta^J\mathcal{A} = A_{j_1}\Delta A_{j_2}\Delta \cdots\Delta A_{j_l}$. Then by definition of $k$-ary symmetric difference of sets we have the following proposition:

\begin{pro}\label{Symmetric_Difference}
In terms of the sets $A^I$, where $\emptyset \neq I\subseteq [k]$, $\Delta^J\mathcal{A}$ can be expressed as
$$\Delta^J\mathcal{A} = \bigsqcup_{\substack{I\subseteq[k],\\
                  |J\cap I|~is ~odd}} A^I.$$
\end{pro}
\begin{proof}
If $x\in \Delta^J\mathcal{A}$, then $x\in \bigcup_{j\in J}A_j$ and $|\{j\in J \ : \ x\in A_j\}|$ is odd. This implies that $x\in \bigcap_{j\in J'}A_j$, where $J' = \{j\in J \ : \ x\in A_j\}$. By equation (\ref{intersection in terms of AI}) we can write $\bigcap_{j\in J'}A_j=\bigsqcup_{\substack{I\subseteq[k],\\J'\subseteq I}} A^I $. Since $\bigsqcup_{\substack{I\subseteq[k],\\J'\subseteq I}} A^I \subseteq \bigsqcup_{\substack{I\subseteq[k],\\|J\cap I|~is~odd}} A^I$, we get
$$\Delta^J\mathcal{A} \subseteq \bigsqcup_{\substack{I\subseteq[k],\\|J\cap I|~is~odd}} A^I.$$
On the other hand, if $x\in \bigsqcup_{\substack{I\subseteq[k],\\|J\cap I|~is~odd}} A^I$, then $x\in A^I$ for some $I\subseteq [k]$ such that $|J\cap I|$ is odd. By the definition of the sets $A^I$'s, this means that  $x\in \bigcap_{j\in J\cap I}A_j$ and $x\notin A_j$ for $j\in J\backslash I$. So the set $\{j\in J \ : \ x\in A_j\} = J\cap I$ has odd number of elements. Therefore, $x\in \Delta^J\mathcal{A}$. So
$$\bigsqcup_{\substack{I\subseteq[k],\\|J\cap I|~is~odd}} A^I \subseteq \Delta^J\mathcal{A}, $$
and hence we get
$$\Delta^J\mathcal{A} = \bigsqcup_{\substack{I\subseteq[k],\\|J\cap I|~is~odd}} A^I.$$
\end{proof}
\begin{pro}\label{symmetricdif} Let $C$ be a binary linear code and   $c_1,\ldots,c_k \in C$ linearly independent codewords with supports  $ A_1, A_2,\ldots, A_k$.  Then we have
\begin{eqnarray*}
&&{\rm Supp}(c_1+c_2+\ldots+c_k) = \bigsqcup_{\substack{ I\subseteq[k],\\ |I| \mbox{ \scriptsize{is odd}}}} A^I
\end{eqnarray*}
Equivalently,
$${\rm wt}(c_1+c_2+\ldots+c_k) = \sum_{\substack{I\subseteq[k],\\|I| \mbox{ \scriptsize{is odd}}}}| A^I|.$$
\end{pro}
\begin{proof} As $c_1,\ldots,c_k \in C$ are linearly independent we have $\sum_{i=1}^l c_i\neq 0$ for any $l\in \{1,\ldots,k\}$. Hence ${\rm Supp}(\sum_{i=1}^l c_i)\neq \emptyset$. We proceed the proof by induction on $k$. For $k=2$, we observe that ${\rm Supp}(c_1+c_2) =  A_1 \Delta  A_2$. Then by induction argument and associativity of symmetric difference we get ${\rm Supp}(\sum_{i=1}^{k} c_i)= {A_1\,\Delta\, A_2\,\Delta \ldots\Delta\, A_{k}}$. By Lemma \ref{oddindex} we get that
\begin{eqnarray*}
&&{\rm Supp}(c_1+c_2+\ldots+c_k) = \bigsqcup_{\substack{I\subseteq[k],\\|I| \mbox{ \scriptsize{is odd}}}} A^I
\end{eqnarray*}
Equivalent statement follows easily, as the sets $A^I$'s are pairwise disjoint.
\end{proof}
\section{A characterization of constant weight codes}
\label{Sec:4}
In this section, we give a characterization for the binary linear constant weight codes. As a consequence of this characterization we get the relations between the parameters of a given binary linear constant weight code.

\begin{pro}\label{intersection of Supports} Let $C$ be a constant weight code of weight $w$. If $c_1,c_2,\ldots,c_k$ are linearly independent codewords with the supports $ A_1, A_2,\ldots, A_k$ respectively, then we have
\begin{equation}\label{First_Condition}
\left|\bigcap_{i=1}^k A_i\right|=\frac{w}{2^{k-1}}.
\end{equation}
\end{pro}
\begin{proof}
We  use the induction on $k$. For $k=2$, it follows from Lemma \ref{intersectweight}. Suppose it is true for $l<k$, and assume $c_1,c_2,\ldots,c_k$ are linearly independent codewords with supports $ A_1, A_2,\ldots, A_k$ in a constant weight code $C$ of weight $w$. By induction hypothesis we have, $| A_1\cap A_2\cap\ldots\cap A_{k-1}| = \dfrac{w}{2^{k-2}}$. Note that, by Proposition \ref{intersections}, we have $ A_1\cap A_2\cap\ldots\cap A_{k-1} =  A^{12\ldots(k-1)}\sqcup A^{12\ldots k}$. Now, let $x=| A^{12\ldots k}|$. Then, $| A^{12\ldots(k-1)}|= \dfrac{w}{2^{k-2}}-x $. Similarly, for any subset $I=\{i_1,i_2,\ldots, i_{k-1}\}$ of size $k-1$ of $[k]=\{1,2,\ldots,k\}$, again by Proposition \ref{intersections} we can write $\bigcap_{i\in I} A_i = A^I\sqcup A^{12\ldots k}$. So we obtain $| A^{i_1i_2\ldots i_{k-1}}|= \dfrac{m}{2^{k-2}}-x $ for any $\{i_1,i_2,\ldots, i_{k-1}\}\subseteq [k]$ .

Consider the subsets $I$  of $[k]$ which have  $k-2$ elements and start with the simplest one $I=\{1,2,\ldots,k-2\}$. From Proposition \ref{intersections} we  have
$$ A_1\cap A_2\cap\ldots\cap A_{k-2} =  A^{12\ldots(k-2)}\sqcup A^{12\ldots(k-2)(k-1)}\sqcup A^{12\ldots(k-2)k}\sqcup A^{12\ldots k}$$
Comparing the sizes of both sides and using the induction hypothesis we get,
$$\dfrac{w}{2^{k-3}} = | A^{12\ldots(k-2)}| + \dfrac{w}{2^{k-2}}-x+\dfrac{w}{2^{k-2}}-x+x.$$
So we have  $| A^{12\ldots(k-2)}|=x$. In a similar way, for any $I=\{i_1,i_2,\ldots, i_{k-2}\}\subset [k]$, we have $| A^{i_1i_2\ldots i_{k-2}}|=x$. Continuing like this we  have $| A^{i_1i_2\ldots i_{k-3}}|=\dfrac{w}{2^{k-2}}-x$, $| A^{i_1i_2\ldots i_{k-4}}|=x$, $| A^{i_1i_2\ldots i_{k-5}}|=\dfrac{w}{2^{k-2}}-x$, $| A^{i_1i_2\ldots i_{k-6}}|=x$ \ldots.
At the end of this process, for any $i\in[k]$ we have
$$\mid A^i\mid=\begin{cases}
x, & \mbox{if $k$ is odd}\\
\dfrac{w}{2^{k-2}}-x,& \mbox{if $k$ is even}
\end{cases}
$$
Now, consider the support of the sum of $c_1,\ldots,c_k$. By Proposition \ref{symmetricdif}, we have
$${\rm Support}(c_1+c_2+\cdots+c_k)=  A_1\Delta A_2\Delta\ldots\Delta A_k=\bigsqcup_{\substack{I\subseteq[k],\\|I| \mbox{ \scriptsize{is odd}}}} A^I$$
So, since we are working on a constant weight code of weight $w$, and the number of odd size subsets of $[k]$ is $2^{k-1}$, we get
$$w={\rm wt}(c_1+c_2+\ldots+c_t) = \sum_{\substack{I\subseteq[k],\\|I| \mbox{ \scriptsize{is odd}}}}| A^I|=\begin{cases}
2^{k-1}x, & \mbox{if $k$ is odd}\\
2^{k-1}(\dfrac{w}{2^{k-2}}-x),& \mbox{if $k$ is even}
\end{cases}$$
Hence, in any case we have the desired result:
$$| A_1\cap A_2\cap\ldots\cap A_{k}| = | A^{12\ldots k} | =x=\dfrac{w}{2^{k-1}}$$
\end{proof}

In the above proof, we see that if $C$ is a constant weight code of weight $w$ and supports of linearly independent codewords in $C$ satisfy the condition (\ref{First_Condition}), then the sets in the partition defined above have the same number of elements. Therefore, we have the following corollary.
\begin{cor}
Let $C$ be a $k$-dimensional binary constant weight code of weight $w$. If the codewords $c_1,c_2,\ldots,c_k$ with supports $ A_1, A_2,\ldots, A_k$ form a basis for $C$, then $|A^I|=\dfrac{w}{2^{k-1}}$ for any non-empty subset $I\subseteq [k]$.
\end{cor}
\begin{cor}\label{atleastweight} Let $C$ be a $k$-dimensional binary constant weight code. Then weight of $C$ is a multiple of $2^{k-1}$.
\end{cor}
\begin{proof}
It follows from the fact that the number \(\left|\bigcap_{i=1}^k A_i\right|=\frac{w}{2^{k-1}} \) is a positive integer.
\end{proof}
Now we can easily conclude the following result.
\begin{cor}\label{dim3} Let $k\geq 3$ be an integer. There is no $k$-dimensional constant weight code of weight $2$.
\end{cor}
As a corrollary we get the explicit relations between the parameters of a given binary linear constant weight codes which are also obtained in \cite{Bonisoli} by using Plotkin bound. Our calculations do not use  Plotkin bound, they just follows from the intersection of the supports of the codewords (see Proposition \ref{intersection of Supports}).
\begin{cor}\label{atleastlength}
Let $C$ be a binary $k$-dimensional constant weight code such that weight of $C$ is $2^{k-1}m$ for some positive integer $m$. Then the length of $C$ is at least $(2^k-1)m.$
\end{cor}
\begin{proof}
From the proof of the Proposition \ref{intersection of Supports} we know that if the weight of a $k$-dimensional constant weight  code is $w=2^{k-1}m$ for some positive integer $m$, then for any non-empty subset $I\subseteq [k]$, $| A^I|=\dfrac{w}{2^{k-1}}=m$. Since there are exactly $2^k-1$ many such subsets $I$'s and the corresponding sets $ A^I$'s are disjoint subsets of the set $\{1,2,\ldots,n\}$, where $n$ is the length of codewords, the minimum length $n$ is at least $(2^k-1)m$.
\end{proof}

Conversely, we have the following theorem.
\begin{thm}\label{converse}
Suppose $c_1,c_2,\ldots,c_k$ are linearly independent codewords of the same weight $w$ with supports $A_1,A_2,\ldots,A_k$. If every set in the collection $\{A^I \ | \ I\subseteq [k], \ |I|>0\}$ has $\dfrac{w}{2^{k-1}}$ elements , then the codewords $c_1,c_2,\ldots,c_k$ generate  a $k$-dimensional constant weight  code of weight $w$.
\end{thm}
\begin{proof}
Suppose $c_1,c_2,\ldots,c_k$ are linearly independent codewords of the same weight $w$ with supports $A_1,A_2,\ldots,A_k$ satisfying the condition $|A^I|=\dfrac{w}{2^{k-1}}$ for all non-empty subset $I$ of $[k]$. Now, in order to show that the linearly independent codewords $c_1,c_2,\ldots,c_k$ generate a $k$-dimensional constant weight  code, it is enough to show that the weight of the sum $c_{i_1}+c_{i_2}+\ldots+c_{i_l}$ is $w$, for any subset $J=\{i_1,i_2,\ldots, i_{l}\}\subseteq [k]$. Since the support of the sum $c_{i_1}+c_{i_2}+\ldots +c_{i_l}$ is $A_{i_1}\Delta A_{i_2}\Delta \cdots\Delta A_{i_l}$, by Proposition \ref{Symmetric_Difference} we have
$$\mbox{Supp}(c_{i_1}+c_{i_2}+\ldots+c_{i_l})=\Delta^J\mathcal{A} = \bigsqcup_{\substack{I\subseteq[k],\\|J\cap I| \scriptsize{\mbox{ is odd}}}} A^I,$$
which implies that
$$\mbox{wt}(c_{i_1}+c_{i_2}+\ldots+c_{i_l})= \sum_{\substack{I\subseteq[k],\\|J\cap I| \scriptsize{\mbox{ is odd}}}} |A^I|.$$
By assumption we have $|A^I|=\dfrac{w}{2^{k-1}}$ for all $\emptyset \neq I\subseteq [k]$. Since the number of subsets $I\subseteq[k]$ satisfying the condition ``$|J\cap I|$ is odd'' is $2^{k-1}$, we obtain the required result
$$\mbox{wt}(c_{i_1}+c_{i_2}+\ldots+c_{i_l})= \sum_{\substack{I\subseteq[k],\\|J\cap I|=\scriptsize{\mbox{odd}}}} |A^I|=2^{k-1}\dfrac{w}{2^{k-1}}=w.$$
\end{proof}
\begin{rem}
Note that, the Proposition \ref{intersection of Supports} and the Theorem \ref{converse} together give us a  correspondence between $k$-dimensional constant weight codes of weight $w$ and length $n$, and the partitions $ \{ A^I \ : \ I\subseteq [k], |I|>0 \}$ of subsets of $ \{ 1,2,\ldots, n\}$ with size $(2^k-1)\dfrac{w}{2^{k-1}}$ satisfying the condition $|A^I|=\dfrac{w}{2^{k-1}}$ for all non-empty $I\subseteq [k]$.
\end{rem}
\section{A construction of constant weight codes up to equivalence}
\label{Sec:5}
In the next section, we consider an algebraic approach for the binary linear constant weight codes in terms of their automorphism groups. To determine the permutation automorphism groups of the codes, we need a list of explicit basis elements.
\begin{thm}\label{construction} For any positive integer $m$, the codewords of length $(2^k-1)m$ defined as
\begin{eqnarray*}
&&c_1=(\overbrace{1,1,\ldots,1}^{2^{k-1}m},\overbrace{0,0,\ldots,0}^{(2^{k-1}-1)m}), \quad c_2=(\overbrace{1,\ldots,1}^{2^{k-2}m},\overbrace{0,\ldots,0}^{2^{k-2}m},\overbrace{1,\ldots,1}^{2^{k-2}m},\overbrace{0,\ldots,0}^{(2^{k-2}-1)m}), \\
&& c_3=(\overbrace{1,\ldots,1}^{2^{k-3}m},\overbrace{0,\ldots,0}^{2^{k-3}m},\overbrace{1,\ldots,1}^{2^{k-3}m},\overbrace{0,\ldots,0}^{2^{k-3}m},\overbrace{1,\ldots,1}^{2^{k-3}m},\overbrace{0,\ldots,0}^{2^{k-3}m},\overbrace{1,\ldots,1}^{2^{k-3}m},\overbrace{0,\ldots,0}^{(2^{k-3}-1)m}),\\
&& \vdots\\
&& c_{k-1} = (\overbrace{1,\ldots,1}^{2m},\overbrace{0,\ldots,0}^{2m},\ldots,\overbrace{1,\ldots,1}^{2m},\overbrace{0,\ldots,0}^{2m},\overbrace{1,\ldots,1}^{2m},\overbrace{0,\ldots,0}^{m}),\\
&& c_k = (\overbrace{1,\ldots,1}^{m},\overbrace{0,\ldots,0}^{m},\overbrace{1,\ldots,1}^{m},\overbrace{0,\ldots,0}^{m},\ldots,\overbrace{1,\ldots,1}^{m},\overbrace{0,\ldots,0}^{m},\overbrace{1,\ldots,1}^{m}),
\end{eqnarray*}
generate a $k$-dimensional constant weight code of weight $2^{k-1}m$.
\end{thm}
\begin{proof}
Clearly, the described codewords are linearly independent (because, $c_1+\ldots+c_t\neq0$). By construction they satisfy the criteria of the Theorem \ref{converse}. So they generate a $k$-dimensional constant weight  code of weight $2^{k-1}m$:
\end{proof}
\begin{rem}\label{explanation}
Note that, by adding the same number of zeroes at the end of each codeword described in  the previous theorem, one can obtain a constant weight  code of arbitrary length $n>(2^k-1)m$. In \cite{Bonisoli},  Bonisoli proves that two $q$-ary linear constant weight codes (linear equidistant codes) with the same parameters are equivalent according to the definition of \cite{MacWilliams}. In the binary case this equivalence is the same as permutation equivalence. For the proof he starts with a generator matrix of a linear equidistant code and uses an induction argument on the dimension of the code to show that the generator matrix is equivalent to a block matrix formed by Hamming matrices which generates a linear equidistant code.  For  Bonisoli's result on permutation equivalence, we give a simpler proof by using the characterization in Theorem \ref{converse} and  by  writing an explicit permutation of the supports of the codeswords which gives the permutation equivalence. As a conclusion,  Theorem \ref{construction} gives all $k$-dimensional constant weight codes of weight $2^{k-1}m$, ($m\in \mathbb{N}$), and length $n\ge (2^k-1)m$ up to permutation equivalence.
\end{rem}
\begin{thm}[\cite{Bonisoli}]\label{Perm_Equiv}
If $C_1$ and $C_2$ are constant weight codes of the same dimension, weight and length, then they are permutation equivalent.
\end{thm}
\begin{proof}
For each $k$-dimensional constant weight code of weight $w$ and length $n$, fix a basis $c_1,c_2,\ldots,c_k$, and denote their supports by $A_1,A_2,\ldots,A_k$, respectively. By Theorem \ref{converse}, the constant weight code $C$ is completely characterized  by the partition $\{A^I \ : \ \emptyset\neq I \subseteq [k]\}$ of $\cup_{i=1}^kA_i$. Now, suppose that $C_1$ and $C_2$ are two $k$-dimensional constant weight codes of the same weight $w=2^{k-1}m$, and of the same length $n\geq (2^k-1)m$. Assume $\{A^I \ : \ \emptyset\neq I \subseteq [k]\}$ and $\{B^I \ : \ \emptyset\neq I \subseteq [k]\}$ are the corresponding partitions. Since $A^I$ and $B^I$ are subsets of $\{1,2,\ldots,n\}$ and they have the same number of elements, we can find a permutation $\sigma_I\in S_n$ which maps $A^I$ to $B^I$ and fixes the complement of $A^I\cup B^I$. Indeed, if $A^I=\{r_1,r_2,\ldots,r_m\}$ and $B^I=\{s_1,s_2,\ldots,s_m\}$, where $m=\dfrac{w}{2^{k-1}}$, then the permutation $\sigma_I = (r_1s_1)(r_2s_2)\cdots(r_ms_m)$ is the required permutation. Then the permutation $\sigma = \prod_{\emptyset \neq I\subseteq[k]}\sigma_I$ maps the partition $\{A^I \ : \ \emptyset\neq I \subseteq [k]\}$ to the partition $\{B^I \ : \ \emptyset\neq I \subseteq [k]\}$. Hence, it maps $C_1$ to $C_2$.
 \end{proof}
\section{Permutation automorphism groups of constant weight codes} \label{Sec:6}
Let $S_n$ be the symmetric group on $n$ letters. There is an action of $S_n$ on the vector space $\mathbb{F}_2^n$ given as follows: Let $\sigma\in S_n$. Then for any $v=(v_1,v_2,\ldots,v_n)$, $\sigma v=(v_{\sigma^{-1}(1)}, v_{\sigma^{-1}(2)},\ldots,v_{\sigma^{-1}(n)})$.
The stabilizer of $C$ under this action is $\PAut(C)=\{\sigma\in S_n \ :\  \sigma(C)=C\}$. It is clear that $\PAut(C)$ is a subgroup of $S_n$. It is called the permutation automorphism group of $C$. If two codes are permutation equivalent, corresponding permutation automorphism groups are isomorphic. In fact, whenever $\sigma C_1=C_2$, the corresponding permutation automorphism groups are conjugate, i.e., we have $\PAut(C_2)=\PAut(C_1)^{\sigma}$.
In this section, we consider the permutation automorphism group of a constant weight code.

Consider the constant weight code %
$$C = \{(0,0,\ldots,0),(1,1,0,\ldots,0),(1,0,1,\ldots,0),(0,1,1,\ldots,0)\}.$$
We have proved that any $2$-dimensional constant weight codes of length $n$ and of weight $2$ is permutation equivalent to $C$. So, for a $2$-dimensional constant weight code of weight $2$ and length $n$, it is enough to consider permutation automorphism group of the code $C$ given above. By construction, it is clear that $\PAut(C)\cong H\times K$, where $H$ and $K$ are the symmetric groups on $3$-letters $\{1,2,3\}$ and $(n-3)$-letters $\{4,5,\ldots,n\}$, respectively. So the order of $\PAut(C)$ is $3!(n-3)!$. This proves the following proposition:
\begin{pro}\label{weight2} Let $C$ be a $2$-dimensional constant weight code of weight $2$ and length $n\ge 3$. Then $\PAut(C)\cong S_3\times S_{n-3}$.
\end{pro}
To find the isomorphism type of the permutation automorphism groups of constant weight codes of larger weight, we calculate the total number of distinct $k$-dimensional constant weight codes. Moreover, we find a formula for the order of permutation automorphism group of a given constant weight code.
\subsection{The order of permutation automorphism groups}
 We prove that the weight of a $k$-dimensional constant weight code must be of the form $2^{k-1}m$ for some positive integer $m$ and its length $n$ must satisfy the condition $n\ge (2^k-1)m$. We also prove that there is a  correspondence between $k$-dimensional constant weight codes with weight $2^{k-1}m$ and length $n\ge (2^k-1)m$, and the families $\{A^I\subseteq [n]\ : \ I\subseteq [k], |I|>0 \mbox{ and } |A^I|=m \}$.

The number of such families can be obtained by the following formula
\begin{equation}
{n \choose (2^k-1)m}\prod_{i=1}^{2^k-1}{(2^k-i)m \choose m}
\end{equation}
Since the order is not important we have to divide this number by $k!$. On the other hand, the number of unordered bases of a $k$-dimensional vector space over $\mathbb{F}_2$ is $\dfrac{\prod_{i=0}^{k-1}(2^k-2^i)}{k!}$. Hence, the total number of distinct $k$-dimensional constant weight  codes with  weight $2^{k-1}m$ and length $n\ge (2^k-1)m$ is
\begin{equation} \label{number_of_codes1}
\frac{{n \choose (2^k-1)m}\displaystyle\prod_{i=1}^{2^k-1}{(2^k-i)m \choose m}}{\displaystyle\prod_{i=0}^{k-1}(2^k-2^i)}
\end{equation}
In terms of multinomial numbers the above formula can be written as
\begin{equation}\label{number_of_codes2}
\frac{{n \choose (2^k-1)m}{(2^k-1)m \choose m,m,\ldots,m}}{\displaystyle\prod_{i=0}^{k-1}(2^k-2^i)},
\end{equation}
where
$${(2^k-1)m \choose \underbrace{m,m,\ldots,m}_{2^k-1}}=\frac{((2^k-1)m)!}{m!m!\ldots m!}=\frac{((2^k-1)m)!}{(m!)^{2^k-1}}$$

The action of $S_n$ on $\mathbb{F}_2^n$ is linear. For two permutations $\sigma_1,\sigma_2\in S_n$, if $\sigma_1^{-1}\sigma_2\in \PAut(C)$, then $\sigma_1^{-1}\sigma_2(C)=C$. So $\sigma_1(C)=\sigma_2(C)$ for any two permutations in the same left cosets of $\PAut(C)$. Therefore, if $\sigma_1, \sigma_2$ are representatives of two distinct left cosets of $\PAut(C)$, then $\sigma_1(C)\neq\sigma_2(C)$. Hence, the number of distinct constant weight codes of the same parameters as of $C$ is the same as the number of left cosets of $\PAut(C)$ in $S_n$. Thus, by the formula (\ref{number_of_codes2}) the order of $\PAut(C)$ can be given
\begin{equation}\label{order_of_PAut 1}
|\PAut(C)|=  \frac{n!\displaystyle\prod_{i=0}^{k-1}(2^k-2^i)}{{n \choose (2^k-1)m}{(2^k-1)m \choose m,m,\ldots,m}}.
\end{equation}
Simplifying the right hand side, we obtain the following result.

\begin{thm}\label{order_of_PAut 2} Let $C$ be a $k$-dimensional constant weight code of weight $2^{k-1}m$ and length $n$. Then we have

 \begin{equation}
|\PAut(C)|=  (n-(2^k-1)m)!(m!)^{2^k-1}\prod_{i=0}^{k-1}(2^k-2^i) .
\end{equation}

\end{thm}
Using a simple induction on $k$, one can easily show that the product $\prod_{i=0}^{k-1}(2^k-2^i)$ is always divisible by $6$ for $k\ge 2$. Hence, we have the following result.
\begin{cor}\label{Order of PAut}
If $C$ is a constant weight  code of dimension $\ge2$, then the order of its permutation automorphism group is a multiple of $6$.
\end{cor}

There are $1$-dimensional constant weight codes so that $\PAut(C)\cong C_2$ and $\PAut(C)\cong C_2\times C_2$, for example for  $C=\{0000,1100\}$ we have $\PAut(C)\cong C_2\times C_2$ and for $C=\{000,100\}$ has $\PAut(C)\cong C_2$.
For larger dimensional codes, it is quite hard to determine the isomorphism type of $\PAut(C)$ of a given constant weight code $C$. We also prove the following proposition.
\begin{pro}\label{weight4} Let $C$ be a $2$-dimensional constant weight code of weight $4$ and length $n\ge 6$. Then $\PAut(C)\cong S_2\times S_4\times S_{n-6}$.
\end{pro}
\begin{proof}
It is easy to see that $D=\{000000, 111100, 001111, 110011\}$ is a $2$-dimensional constant weight code of weight $4$ and length $6$. By Theorem \ref{Perm_Equiv}, any constant weight code of the same parameters is permutation equivalent to $D$ and they have isomorphic permutation automorphism group.  So for any $2$-dimensional constant weight code $C$ of weight $4$ and length $6$ , $\PAut(C)\cong \PAut(D)$. So for such codes it is enough to compute $\PAut(D)$. Using the computer program SAGE \cite{sage},  we  compute $\PAut(D)$ and show that it is isomorphic to  $ S_2\times S_4$. Indeed,  the permutations $\sigma_1=(145236)$ and $\sigma_2=(1324)(56)$ map $D$ to itself. So the group $G=\langle \sigma_1,\sigma_2\rangle$ generated by these permutations, which is also isomorphic to $ S_2\times S_4$, is a subgroup of $\PAut(D)$. By the formula (\ref{number_of_codes1}) or (\ref{number_of_codes2}), the number of distinct $2$-dimensional  constant weight codes of weight $4$ and length $6$ is  $15$, which is also the index of $\PAut(D)$ in the symmetric group $S_6$. So the order of $\PAut(D)$ is $6!/15 =48$. Since the order of $G$ is also $48$ and $G\le \PAut(D)$, $\PAut(D) = G$. Now, if the length $n>6$, then by adding $n-6$ zeroes to the end of each codeword in $D$, we can extend $D$ to a $2$-dimensional constant weight code $\widetilde{D}$ of weight $4$ and length $n>6$. And by Theorem \ref {Perm_Equiv}, any $2$-dimensional constant weight code of weight $4$ and length $n$ is permutation equivalent to $\widetilde{D}$. By the same arguments as above, one can easily show that $\PAut(\widetilde{D}) = \PAut(D)\times H$, where $H$ is the symmetric group on $n-6$ letters $\{7,8,\ldots n\}$. Hence, we get $\PAut(\widetilde{D})\cong S_2\times S_4\times S_{n-6}$.
\end{proof}

\begin{thm}Let $C$ be a two-dimensional binary linear constant weight code of weight $2m$ and length $3m$ where $m$ is a positive integer. Then $\PAut(C)\cong S_m\wr S_3$, the wreath product of $S_m$ by $S_3$.

\end{thm}

\begin{proof} Let $D$ be the code spanned by the set  $$\{c_1,c_2\} = \{(\overbrace{1,1,\ldots,1}^{2m},\overbrace{0,0,\ldots,0}^{m}), \quad (\overbrace{1,\ldots,1}^{m},\overbrace{0,\ldots,0}^{m},\overbrace{1,\ldots,1}^{m})\}$$
By Theorem \ref{construction}, it is clear that $D$ is a constant weight code. Any $2$-dimensional binary constant weight of weight $2m$ and length $3m$ is permutation equivalent to $D$ by Theorem \ref{Perm_Equiv}. As permutation equivalent codes have isomorphic permutation automorphism groups, it is enough to calculate $PAut(D)$.  Clearly, $S_m\times S_m\times S_m$ is a subgroup of $\PAut(D)$ as each block of $m$-zeroes or $m$-ones is fixed under the action of $S_m$.  If we permute the three blocks of $c_2$, we also get either $c_1$, $c_2$ or $c_1+c_2$. Therefore, the wreath product $S_m\wr S_3$ of $S_m$ by $S_3$ is also a subgroup of $\PAut(D)$. By Theorem (\ref{order_of_PAut 2}), we have $\mid \PAut(D)\mid  = 6(m!)^3$. As $\mid S_m\wr S_3\mid$ is also $6(m!)^3$, we conclude that $\PAut(D) = S_m\wr S_3$.\end{proof}

Recall that the wreath product $S_m\wr S_3$ is the split extension $(S_m\times S_m\times S_m )\rtimes S_3$. For the details about the wreath product, one can see [\cite{Leedham}, page 38].

\begin{rem} If we consider the code $C$ with a basis $c_1,c_2,\ldots,c_k$ given in Theorem \ref{construction} we can observe that whenever $m\geq 3$, the group $G = \langle(123),(12)\rangle$ is a subgroup of $\PAut(C)$. But, it is not true for the cases $m=1$ and $m=2$, that is, the group $G$ does not fix the code $C$.  However, for these cases we can construct explicitly the subgroups of $\PAut(C)$ which are isomorphic to $S_3$. For the first case, where $m=1$, note that, the supports of the last two basis elements $c_{k-1}$ and $c_k$ are of the form $A_{k-1}=\{1,2,5,6,9,10,\ldots,2^k-3,2^k-2\}$ and $A_{k}=\{1,3,5,7,9,11,\ldots,2^k-3,2^k-1\}$, respectively. Clearly, the permutations $$\sigma_1=(23)(67)(10\;11)\ldots(2^k-2\;2^k-1)$$
$$\sigma_2=(123)(567)(9\;10\;11)\ldots(2^k-3\;2^k-2\;2^k-1)$$ fix the code $C$. So the group $G_1$ generated by $\sigma_1$ and $\sigma_2$ is a subgroup of $\PAut(C)$. Since $G_1$ is not abelian and it has order 6, it is isomorphic to $S_3$.

For the latter case ($m=2$), the supports of the last two basis elements are of the form
$$A_{k-1}=\{1,2,3,4,9,10,11,12,\ldots,2^{k+1}-7,2^{k+1}-6,2^{k+1}-5,2^{k+1}-4\} \mbox{ and } $$ $$ A_{k}=\{1,2,5,6,9,10,\ldots,2^{k+1}-3,2^{k+1}-2\}.$$
By direct calculations one can show that the following permutations
$$\tau_1=(35)(46)(11\;13)(12\;14)\ldots(2^{k+1}-5\;\ 2^{k+1}-3)(2^{k+1}-4\;\ 2^{k+1}-2) $$
$$\tau_2=\prod_{i=2}^k(2^{i+1}-7\;\ 2^{i+1}-5\;\ 2^{i+1}-3)
(2^{i+1}-6\;\ 2^{i+1}-4\;\ 2^{i+1}-2)$$
fix the code $C$. So the group $G_2 = \langle\tau_1,\tau_2\rangle$, is a subgroup of $\PAut(C)$. Again, since $G_2$ is non-abelian and has order 6, it is isomorphic to $S_3$. Hence, we just proved the following proposition.
\end{rem}
\begin{pro} Let $C$ be a constant weight code of dimension $\geq 2$. Then $PAut(C)$ has a subgroup isomorphic to $S_3$.
\end{pro}

By Proposition \ref{Order of PAut} we have $|\PAut(C)|\ge 6$, for any constant weight code of dimension $k\ge 2$. So it is natural to ask for which constant weight codes, permutation automorphism group is isomorphic to $S_3$.
\begin{thm} Let $C$ be a  constant weight code of dimension $k\ge 2$. Then $\PAut(C)\cong S_3$ if and only if $k=2$, $wt(C)=2$ and the length of $C$ is either $3$ or $4$.
\end{thm}
\begin{proof} In Proposition \ref{weight2}, we proved that $\PAut(C)\cong S_3$, whenever $C$ is two-dimensional constant weight code of weight $2$ and length $3$ or $4$. Converse follows from the formula in (\ref{order_of_PAut 2}) ($|\PAut(C)| = 6$ only when $k=2$, $m=1$ and $n=3$ or $4$).
\end{proof}

\section{Concluding Remarks} 
Let $G$ be a finite group and $K$ a finite field. The group algebra $KG$ is a $K$-vector space spanned by $G$ with a ring structure on it. A left $G$-code (or group code for $G$)
  is a left ideal $C$ in $KG$. The following theorem is a well known characterization of left $G$-codes.

\begin{thm}\label{groupcode}[\cite{BernalRioSimon}, Theorem 1.2]
Let $G$ be a finite group of order $n$ and let $C$ be a linear code in $K^n$. 
Then $C$ is a left  $G$-code if and only if $G$ is isomorphic to a transitive subgroup of $S_n$ contained in $\PAut(C)$.

\end{thm}
Whenever a linear code is a left $G$-code, there are more algebraic tools to understand such codes (see \cite{BorelloWillems}). By using Proposition \ref{weight2} and Theorem \ref{groupcode},
it is easy to conclude that there is no $2$-dimensional binary linear constant weight code of weight $2$ and length $n>3$ which is a left $G$-code for any group $G$. 
Moreover as we have the formula \begin{equation}\label{order_of_PAut 2}
|\PAut(C)|=  (n-(2^k-1)m)!(m!)^{2^k-1}\prod_{i=0}^{k-1}(2^k-2^i) .
\end{equation} for a $k$-dimensional  binary linear constant weight code of weight $2^{k-1}m$ and the length $n\geq (2^k-1)m$, for some specific parameters there is no constant weight code which is a left $G$-code for any group $G$. For example, there is no $2$-dimensional constant weight code of weight $4$ and length $7$ which is a left $G$-code for any group $G$.

Naturally, one can ask the following questions:

\begin{que}  What are the conditions on the parameters of a binary constant weight code $C$ so that $C$ is a left $G$-code for a group $G$?
\end{que}
More generally, 
\begin{que} Which constant weight codes are left $G$-codes?

\end{que}

\section*{Acknowledgments}  

We would like to thank to Prof. Arrigo Bonisoli, Prof. Giuseppe Mazzuoccolo and Prof. Simona Bonvicini for sending a copy of the paper
 "Every equidistant linear code is a sequence of dual Hamming codes" by A. Bonisoli.
\bibliographystyle{amsplain}

\end{document}